\documentclass[11pt]{article}
\usepackage{hyperref}
\usepackage{algorithm}
\usepackage{algpseudocode}

\usepackage{bbm,amsmath,amsfonts,amsthm,amsopn,amssymb}
\usepackage{color}

\usepackage{graphicx}
\usepackage{epstopdf}
\usepackage{subcaption}

\newtheorem{theorem}{Theorem}[section]

\newtheorem{corollary}[theorem]{Corollary}

\newtheorem{lemma}[theorem]{Lemma}

\newcommand{\BALD}{\begin{aligned}}
\newcommand{\EALD}{\end{aligned}}
\newcommand{\BALDS}{\begin{aligned*}}
\newcommand{\EALDS}{\end{aligned*}}
\newcommand{\BCAS}{\begin{cases}}
\newcommand{\ECAS}{\end{cases}}
\newcommand{\BEAS}{\begin{eqnarray*}}
\newcommand{\EEAS}{\end{eqnarray*}}
\newcommand{\BEQ}{\begin{equation}}
\newcommand{\EEQ}{\end{equation}}
\newcommand{\BIT}{\begin{itemize}}
\newcommand{\EIT}{\end{itemize}}
\newcommand{\BMAT}{\begin{bmatrix}}
\newcommand{\EMAT}{\end{bmatrix}}
\newcommand{\BNUM}{\begin{enumerate}}
\newcommand{\ENUM}{\end{enumerate}}

\newcommand{\ie}{{\it i.e.}}

\newcommand{\BA}{\begin{array}}
\newcommand{\EA}{\end{array}}

\newcommand{\reals}{\mathbf{R}}

\newcommand{\Prob}{\mathop{\mathbf{Pr}}}

\newcommand{\norm}[1]{\left\| #1 \right\|}

\newcommand{\cN}{\mathcal{N}}

\DeclareMathOperator{\unif}{Unif}

\newcommand{\R}{\mathbb{R}}
\newcommand{\V}{{\cal V}}

\newcommand{\Pp}{\mathbb{P}}
\newcommand{\E}{E}
\newcommand{\hbeta}{\hat{\beta}}

\newcommand{\indicator}[1]{\mathbbm{1}\left(#1\right)}

\begin{document}

\title{Exact Post Model Selection Inference for Marginal Screening}
\author{Jason D. Lee and Jonathan E. Taylor}
\maketitle

\begin{abstract} 
We develop a framework for post model selection inference, via marginal screening, in linear regression. At the core of this framework is a result that characterizes the exact distribution of linear functions of the response $y$, conditional on the model being selected (``condition on selection" framework).
This allows us to construct valid confidence intervals and hypothesis tests for regression coefficients that account for the selection procedure. In contrast to recent work in high-dimensional statistics, our results are exact (non-asymptotic) and require no eigenvalue-like assumptions on the design matrix $X$. Furthermore, the computational cost of marginal regression, constructing confidence intervals and hypothesis testing is negligible compared to the cost of linear regression, thus making our methods particularly suitable for extremely large datasets. Although we focus on marginal screening to illustrate the applicability of the condition on selection framework, this framework is much more broadly applicable. We show how to apply the proposed framework to several other selection procedures including orthogonal matching pursuit, non-negative least squares, and marginal screening+Lasso. 
\end{abstract} 

\section{Introduction}
Consider the model 
\begin{align}
y_i = \mu(x_i) + \epsilon_i,\;\epsilon_i \sim\cN(0,\sigma^2 I),
\label{eq:model}
\end{align}
where $\mu(x)$ is an arbitrary function, and $x_i \in \reals^p$. Our goal is to perform inference on $(X^T X)^{-1} X^T \mu$, which is the best linear predictor of $\mu$. In the classical setting of $n>p$ , the least squares estimator 
\begin{align}
\hat \beta = (X^T X)^{-1} X^T y 
\end{align}
is a commonly used estimator for $ (X^T X)^{-1} X^T \mu$. Under the linear model assumption $\mu = X\beta^0$, the exact distribution of $\hat \beta$ is
\begin{align}
\hat \beta \sim \cN(\beta^0, \sigma^2 (X^TX)^{-1}).
\label{eq:ls-dist}
\end{align}
Using the normal distribution, we can test the hypothesis $H_0: \beta^0 _j =0$ and form confidence intervals for $\beta^0_j$ using the z-test.

However in the high-dimensional $p>n$ setting, the least squares estimator is an underdetermined problem, and the predominant approach is to perform variable selection or model selection \cite{buhlmann2011statistics}. There are many approaches to variable selection including AIC/BIC, greedy algorithms such as forward stepwise regression, orthogonal matching pursuit, and regularization methods such as the Lasso. The focus of this paper will be on the model selection procedure known as marginal screening, which selects the $k$ most correlated features $x_j$ with the response $y$.

Marginal screening is the simplest and most commonly used of the variable selection procedures \cite{guyon2003introduction,tusher2001significance,leekasso}. Marginal screening requires only $O(np)$ computation and is several orders of magnitude faster than regularization methods such as the Lasso; it is extremely suitable for extremely large datasets where the Lasso may be computationally intractable to apply. Furthermore, the selection properties are comparable to the Lasso  \cite{genovese2012comparison}. In the ultrahigh dimensional setting $p=O(e^{n^k})$, marginal screening is shown to have the SURE screening property, $P(S \subset \hat S$), that is marginal screening selects a superset of the truly relevant variables \cite{fan2008sure,fan2010sure,fan2009ultrahigh}. Marginal screening can also be combined with a second variable selection procedure such as the Lasso to further reduce the dimensionality; our statistical inference methods extend to the Marginal Screening+Lasso method.

Since marginal screening utilizes the response variable $y$, the confidence intervals and statistical tests based on the distribution in \eqref{eq:ls-dist} are not valid; confidence intervals with nominal $1-\alpha$ coverage may no longer cover at the advertised level:
$$
\Pr\left( \beta^0_j \in C_{1-\alpha}(x) \right)< 1-\alpha.
$$
Several authors have previously noted this problem including recent work in \cite{leeb2003finite,leeb2005model,leeb2006can,berk2013posi}. A major line of work \cite{leeb2003finite,leeb2005model,leeb2006can} has described the difficulty of inference post model selection: the distribution of post model selection estimates is complicated and cannot be approximated in a uniform sense by their asymptotic counterparts. 

In this paper, we describe how to form exact confidence intervals for linear regression coefficients {\it post model selection}. We assume the model \eqref{eq:model}, and operate under the fixed design matrix $X$ setting. The linear regression coefficients constrained to a subset of variables $S$ is linear in $\mu$, $e_j^T(X_S ^T X_S)^{-1} X_S ^T \mu=\eta^T \mu$ for some $\eta$. We derive the conditional distribution of $\eta^T y$ for any vector $\eta$, so we are able to form confidence intervals and test regression coefficients.

In Section \ref{sec:related-work} we discuss related work on high-dimensional statistical inference, and Section \ref{sec:marg-screen} introduces the marginal screening algorithm and shows how z intervals may fail to have the correct coverage properties. Section \ref{sec:selection-event} and \ref{sec:truncated-gaussian-test} show how to represent the marginal screening selection event as constraints on $y$, and construct pivotal quantities for the truncated Gaussian. Section \ref{sec:inference-marg-screen} uses these tools to develop valid hypothesis tests and confidence intervals. 

Although the focus of this paper is on marginal screening, the ``condition on selection" framework, first proposed for the Lasso in \cite{lee2013exact}, is much more general; we use marginal screening as a simple and clean illustration of the applicability of this framework.  In Section \ref{sec:extensions}, we discuss several extensions including how to apply the framework to other variable/model selection procedures and to nonlinear regression problems. Section \ref{sec:extensions} covers
\begin{enumerate}
\item marginal screening+Lasso, a screen and clean procedure that first uses marginal screening and cleans with the Lasso,
\item orthogonal matching pursuit (OMP)
\item non-negative least squares (NNLS).
\end{enumerate}

\section{Related Work}
\label{sec:related-work}
Most of the theoretical work on  high-dimensional linear models focuses on  \emph{consistency}. Such results establish, under restrictive assumptions on $X$, the Lasso $\hat{\beta}$ is close to the unknown $\beta^0$  \cite{negahban2012unified} and selects the correct model \cite{zhao2006model,wainwright2009sharp,lee2013model}. We refer to the reader to \cite{buhlmann2011statistics} for a comprehensive discussion about the theoretical properties of the Lasso.

There is also recent work on obtaining confidence intervals and significance testing for penalized M-estimators such as the Lasso. One class of methods uses sample splitting or subsampling to obtain confidence intervals and p-values \cite{wasserman2009high,meinshausen2009p}. In the post model selection literature, the recent work of  \cite{berk2013posi} proposed the POSI approach, a correction to the usual t-test confidence intervals by controlling the familywise error rate for all parameters in any possible submodel. The POSI approach will produce valid confidence intervals for any possible model selection procedure; however for a given model selection procedure such as marginal regression, it will be conservative. In addition, the POSI methodology is extremely computationally intensive and currently only applicable for $p\le 30$.

A separate line of work establishes the asymptotic normality of a corrected estimator obtained by ``inverting'' the KKT conditions \cite{van2013asymptotically,zhang2011confidence,javanmard2013confidence}. The corrected estimator $\hat{b}$ has the form
$
\hat{b} = \hat{\beta} + \lambda\hat{\Theta}\hat{z},
$
where $\hat{z}$ is a subgradient of the penalty at $\hat{\beta}$ and $\hat{\Theta}$ is an approximate inverse to the Gram matrix $X^TX$. The two main drawbacks to this approach are 1) the confidence intervals are valid only when the M-estimator is consistent, and thus require restricted eigenvalue conditions on $X$, 2) obtaining $\hat{\Theta}$ is usually much more expensive than obtaining $\hat{\beta}$, and 3) the method is specific to regularized estimators, and does not extend to marginal screening, forward stepwise, and other variable selection methods.


Most closely related to our work is the ``condition on selection" framework laid out in \cite{lee2013exact} for the Lasso. Our work extends this methodology to other variable selection methods such as marginal screening,  marginal screening followed by the Lasso (marginal screening+Lasso), orthogonal matching pursuit, and non-negative least squares. The primary contribution of this work is the observation that many model selection methods, including marginal screening and Lasso, lead to ``selection events" that can be represented as a set of constraints on the response variable $y$. By conditioning on the selection event, we can characterize the exact distribution of $\eta^T y$. This paper focuses on marginal screening, since it is the simplest of variable selection methods, and thus the applicability of the ``conditioning on selection event" framework is most transparent. However, this framework is not limited to marginal screening and can be applied to a wide a class of model selection procedures including greedy algorithms such as matching pursuit and orthogonal matching pursuit.  We discuss some of these possible extensions in Section \ref{sec:extensions}, but leave a thorough investigation to future work.

A remarkable aspect of our work is that we only assume $X$ is in general position, and the test is exact, meaning the distributional results are true even under finite samples. By extension, we do not make any assumptions on $n$ and $p$, which is unusual in high-dimensional statistics \cite{buhlmann2011statistics}. Furthermore, the computational requirements of our test are negligible compared to computing the linear regression coefficients. 

Our test assumes that the noise variance $\sigma^2$ is known. However, there are many methods for estimating $\sigma^2$ in high dimensions. A data splitting technique is used in \cite{fan2012variance}, while \cite{sun2012scaled} proposes a method that computes the regression estimate and an estimate of the variance simultaneously. We refer the reader to \cite{reid2013variance} for a survey and comparison of the various methods, and assume $\sigma^2$ is known for the remainder of the paper.

\section{Marginal Screening}
\label{sec:marg-screen}
Let $X\in \R^{n \times p }$ be the design matrix, $y \in \R^n$ the response variable, and assume the model
$$
y_i=\mu (x_i) +\epsilon_i, \epsilon_i \sim \cN(0,\sigma^2 I).
$$ We will assume that $X$ is in general position and has unit norm columns. The algorithm estimates $\hbeta$ via Algorithm \ref{alg:marg-screen}.
\begin{algorithm}
	\caption{Marginal screening algorithm}
\begin{algorithmic}[1]
\State \textbf{Input:} Design matrix $X$, response $y$, and model size $k$.
\State Compute $|X^T y|$.
\State Let $\hat S$ be the index of the $k$ largest entries of $|X^Ty|$.
\State Compute $\hbeta_{\hat S} = (X_{\hat S}^TX_{\hat S})^{-1} X_{\hat S}^T y$
\end{algorithmic}
\label{alg:marg-screen}
\end{algorithm}
The marginal screening algorithm chooses the $k$ variables with highest absolute dot product with $y$, and then fits a linear model over those $k$ variables. We will assume $k\le \min(n,p)$. For any fixed subset of variables $S$, the distribution of $\hbeta_S = (X_S ^T X_S)^{-1} X_S ^T y$ is 
\begin{align}
\hbeta_S &\sim \cN(\beta^\star _{S} , \sigma^2 (X_S ^TX_S)^{-1} )\\
\beta^\star _{S} &:= (X_S^T X_S)^{-1} X_S ^T \mu.
\label{eq:regression-distribution}
\end{align}
We will use the notation $\beta^\star_{j \in S} :=\left(\beta^\star _{S} \right)_j$, where $j$ is indexing a variable in the set $S$.
The z-test intervals for a regression coefficient are 
\begin{align}
&C(\alpha,j,S):=\nonumber\\
&\left( \hbeta_{j \in S} - \sigma z_{1-\alpha/2}  (X_S ^T X_S)_{jj} , \hbeta_{j \in S} + \sigma  z_{1-\alpha/2}  (X_S ^T X_S)_{jj} \right)
\label{eq:z-test-interval}
\end{align}
and each interval has $1-\alpha$ coverage, meaning $\Pr \left(\beta^\star _{j\in S} \in C(\alpha,j,S) \right) =1-\alpha$.
However if $\hat S$ is chosen using a model selection procedure that depends on $y$, the distributional result \eqref{eq:regression-distribution} no longer holds and  the z-test intervals will not cover at the $1-\alpha$ level. It is possible that 
$$
\Pr \left(\beta^\star _{j\in \hat S} \in C(\alpha,j, \hat S) \right) <1-\alpha.
$$
Similarly, the test of the hypothesis $H_0: \beta^\star_{j \in \hat S} =0$ will not control type I error at level $\alpha$, meaning $\Pr \left( \text{reject $H_0$} |H_0 \right) >\alpha$.
\subsection{Failure of z-test confidence intervals}
We will illustrate empirically that the z-test intervals do not cover at $1-\alpha$ when $\hat S$ is chosen by marginal screening in Algorithm \ref{alg:marg-screen}. 
\begin{figure}[h]
\centering
	\includegraphics[width=.5\textwidth]{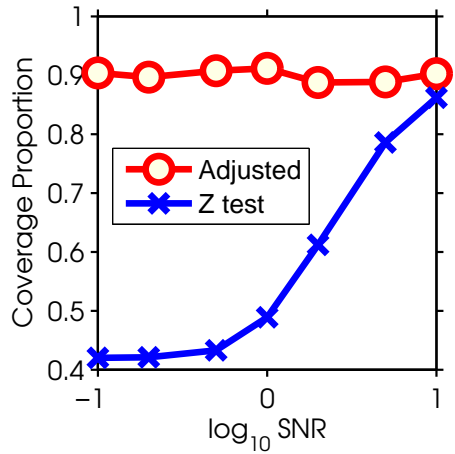}
   \caption{Plots of the coverage proportion across a range of SNR (log-scale). We see that the coverage proportion of the z intervals can be far below the nominal level of $1-\alpha=.9$, even at SNR =5. The adjusted intervals always have coverage proportion $.9$.}
   \label{fig:fail-z-interval}
\end{figure}
For this experiment we generated $X$ from a standard normal with $n=20$ and $p=200$. The signal vector is $2$ sparse with $\beta^0_1, \beta^0_2 = \text{SNR}$,  $y=X\beta^0 +\epsilon$, and $\epsilon \sim N(0,1)$. The confidence intervals were constructed for the $k=2$ variables selected by the marginal screening algorithm. The z-test intervals were constructed via \eqref{eq:z-test-interval} with $\alpha=.1$, and the adjusted intervals were constructed using Algorithm \ref{alg:ci}. The results are described in Figure \ref{fig:fail-z-interval}. The y-axis plots the coverage proportion or the fraction of times the true parameter value fell in the confidence interval. Each point represents $500$ independent trials. The x-axis varies the SNR parameter over the values $0.1, .2, .5, 1, 2, 5, 10$. From the figure, we see that the z intervals can have coverage proportion drastically less than the nominal level of $1-\alpha=.9$, and only for SNR=$10$ does the coverage tend to $.9$. This motivates the need for intervals that have the correct coverage proportion after model selection. 
\section{Representing the selection event}
\label{sec:selection-event}
Since Equation \eqref{eq:regression-distribution} does not hold for a selected $\hat S$ when the selection procedure depends on $y$, the z-test intervals are not valid. Our strategy will be to understand the conditional distribution of $y$ and contrasts (linear functions of $y$) $\eta^T y$, then construct inference conditional on the selection event $\hat E$. We will use $\hat E(y)$ to represent a random variable, and $E$ to represent an element of the range of $\hat E(y)$. In the case of marginal screening, the selection event $\hat E(y)$ corresponds to the set of selected variables $\hat S$ and signs $s$: 
\begin{align}
&\hat E(y) =\left\{y: \text{sign}(x_i ^T y) x_i ^T y > \pm x_j^T y \text{ for all $i \in \hat S$ and $j \in \hat S^c$} \right \} \nonumber\\
&=\left\{ y:\hat s_i x_i^T y > \pm x_j ^T y \text{ and } \hat s_i x_i^T y \ge 0 \text{ for all $i \in \hat S$ and $j \in \hat S^c$}\right\} \nonumber\\
&=\left\{y: A(\hat S,\hat s)y \le b(\hat S,\hat s) \right \}
\label{eq:A-b-defn}
\end{align}
for some matrix $A(\hat S,\hat s)$ and vector $b(\hat S,\hat s)$\footnote{$b$ can be taken to be $0$ for marginal screening, but this extra generality is needed for other model selection methods}. We will use the selection event $\hat E$ and the selected variables/signs pair $(\hat S, \hat s)$ interchangeably since they are in bijection.
%

The space $\reals^n$ is partitioned by the selection events, $$\reals^n = \bigsqcup_{(S,s)} \{y: A(S,s) y \le b(S,s) \}.$$ The vector $y$ can be decomposed with respect to the partition as follows
\begin{align}
y&= \sum_{S,s} y\ \indicator{ A(S,s) y \le b(S,s)}
\end{align}
The previous equation establishes that $y$ is a different constrained Gaussian for each element of the partition, where the partition is specified by a possible subset of variables and signs $(S,s)$. The above discussion can be summarized in the following theorem.
\begin{theorem}
\label{thm:y-cond-dist}
The distribution of $y$ conditional on the selection event is a constrained Gaussian,
\begin{align*}
y|\{\hat E(y) =E\} \overset{d}{=} z \big| \{A(S,s)z \le b\}, \ z\sim \cN(\mu,\sigma^2I).
\end{align*}
\end{theorem}
\begin{proof}
The event $E$ is in bijection with a pair $(S,s)$, and $y$ is unconditionally Gaussian. Thus the conditional $y \big| \{A( S, s)y \le b( S,  s)\}$ is a Gaussian constrained to the set $\{A( S, s)y \le b( S,  s)\}$.
\end{proof}

\section{Truncated Gaussian test}
\label{sec:truncated-gaussian-test}
This section summarizes the recent tools developed in \cite{lee2013exact} for testing contrasts\footnote{A contrast of $y$ is a linear function of the form $\eta^Ty$.} $\eta^Ty$ of a constrained Gaussian $y$. The results are stated without proof and the proofs can be found in \cite{lee2013exact}.

The distribution of a constrained Gaussian $y \sim N(\mu, \Sigma)$ conditional on affine constraints $\{Ay \leq b\}$ has density $ \frac{1}{\Pr(Ay \le b)}f(y;\mu,\Sigma) \indicator{Ay\le b}$, involves the intractable normalizing constant $\Pr(Ay \leq b)$. In this section, we derive a one-dimensional pivotal quantity for $\eta^T\mu$. This pivot relies on characterizing the distribution of $\eta^Ty$ as a truncated normal.
The key step to deriving this pivot is the following lemma:
\begin{lemma}
\label{lem:conditional}
The conditioning set can be rewritten in terms of $\eta^T y$ as follows:
\[  \{Ay \leq b\} = \{\V^-(y) \leq \eta^T y \leq \V^+(y), \V^0(y) \geq 0 \} \]
where 
\begin{align}
\alpha &= \frac{A\Sigma\eta}{\eta^T\Sigma\eta} \label{eq:alpha} \\
\V^- = \V^-(y) &= \max_{j:\ \alpha_j < 0} \frac{b_j - (Ay)_j + \alpha_j\eta^T y}{\alpha_j} \label{eq:v_minus} \\
\V^+ = \V^+(y) &= \min_{j:\ \alpha_j > 0} \frac{b_j - (Ay)_j + \alpha_j\eta^T y}{\alpha_j}. \label{eq:v_plus} \\
\V^0 = \V^0(y) &= \min_{j:\ \alpha_j = 0} b_j - (Ay)_j \label{eq:v_zero}
\end{align}
Moreover, $(\V^+, \V^-, \V^0)$ are independent of $\eta^T y$.
\end{lemma}

The geometric picture gives more intuition as to why $\V^+$ and $\V^-$ are independent of $\eta^T y$. Without loss of generality, we assume $||\eta||_2 = 1$ and $y \sim N(\mu, I)$ (otherwise we could replace $y$ by $\Sigma^{-\frac{1}{2}}y$). Now we can decompose $y$ into two independent components, a 1-dimensional component $\eta^T y$ and an $(n-1)$-dimensional component orthogonal to $\eta$:
\[ y = \eta^T y + P_{\eta^\perp} y. \]
The case of $n=2$ is illustrated in Figure \ref{fig:polytope}. Since the two components are independent, the distribution of $\eta^T y $ is the same as  $\eta^T y| \{P_{\eta^\perp}y \}$. If we condition on $P_{\eta^\perp} y$, it is clear from Figure \ref{fig:polytope} that in order for $y$ to lie in the set, it is necessary for $\V^- \leq \eta^T y \leq \V^+$, where $\V^-$ and $\V^+$ are functions of $P_{\eta^\perp} y$.
\begin{figure}[!h]
\centering
\includegraphics[width = .5\textwidth]{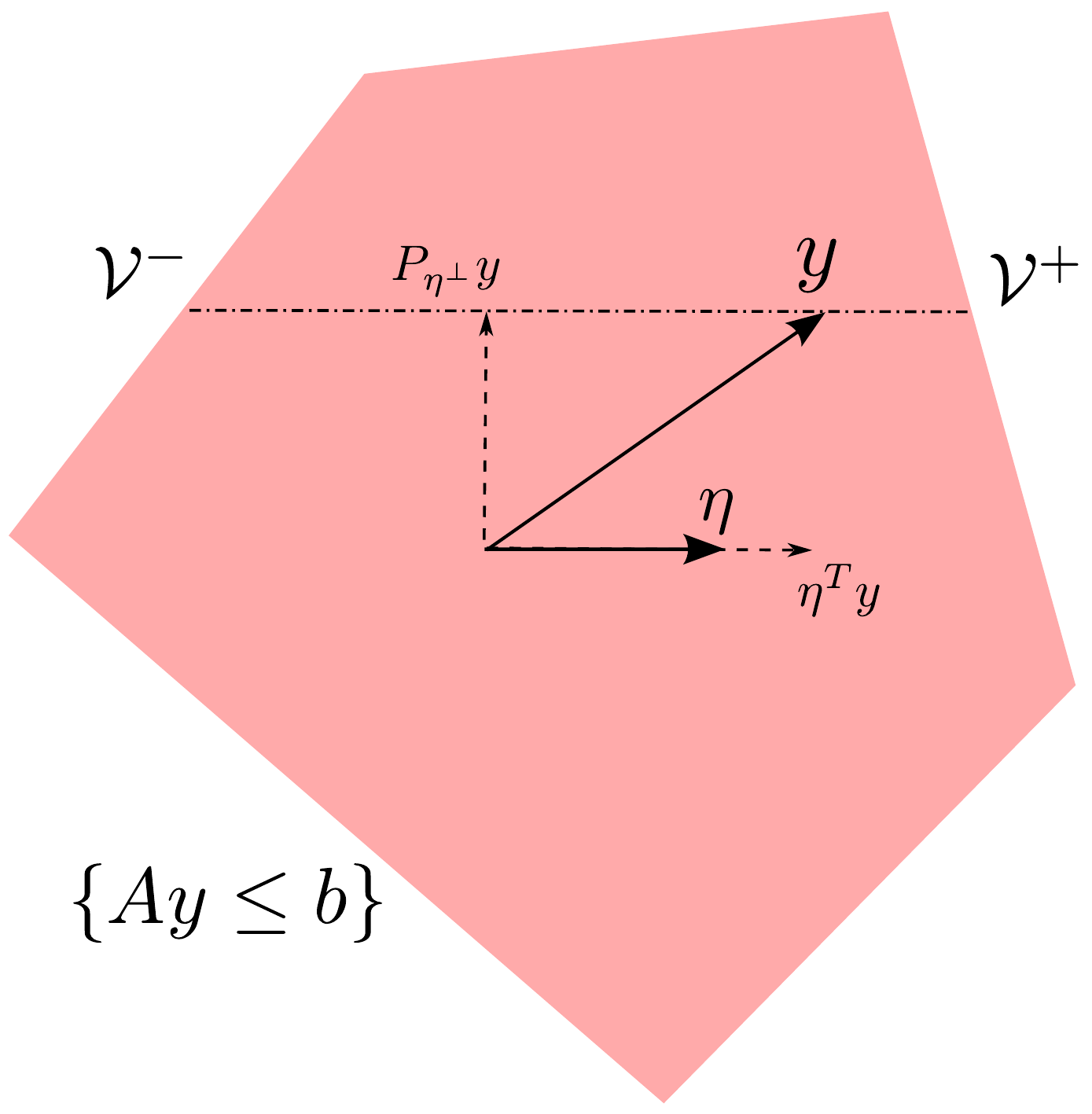}
\caption{A picture demonstrating that the set $\left\{Ay \leq b \right\}$ can be characterized by $\{ \V^- \leq \eta^T y \leq \V^+\}$. Assuming $\Sigma = I$ and $||\eta||_2 = 1$, $\V^-$ and $\V^+$ are functions of $P_{\eta^\perp}y$ only, which is independent of $\eta^T y$.}
\label{fig:polytope}
\end{figure}

\begin{corollary}
The distribution of $\eta^Ty$ conditioned on $\{Ay \leq b, \V^+(y)= v^+, \V^- (y) =v^-\}$ is a (univariate) Gaussian truncated to fall between $\V^-$ and $\V^+$, \ie
\[
\eta^Ty\mid\{Ay \leq b, \V^+(y)= v^+, \V^- (y) =v^-\} \overset{d}{=} W
\]
where $W\sim TN(\eta^T \mu, \eta^T \Sigma\eta ,v^-,v^+)$. $TN(\mu,\sigma,a,b)$ is the normal distribution truncated to lie between $a$ and $b$.

\label{cor:truncated-normal}
\end{corollary}
In Figure \ref{fig:truncated}, we plot the density of the truncated Gaussian, noting that its shape depends on the
 location of $\mu$ relative to $[a,b]$ as well as the width relative to $\sigma$.

\begin{figure}[!h]
\centering
\includegraphics[width = .5\textwidth]{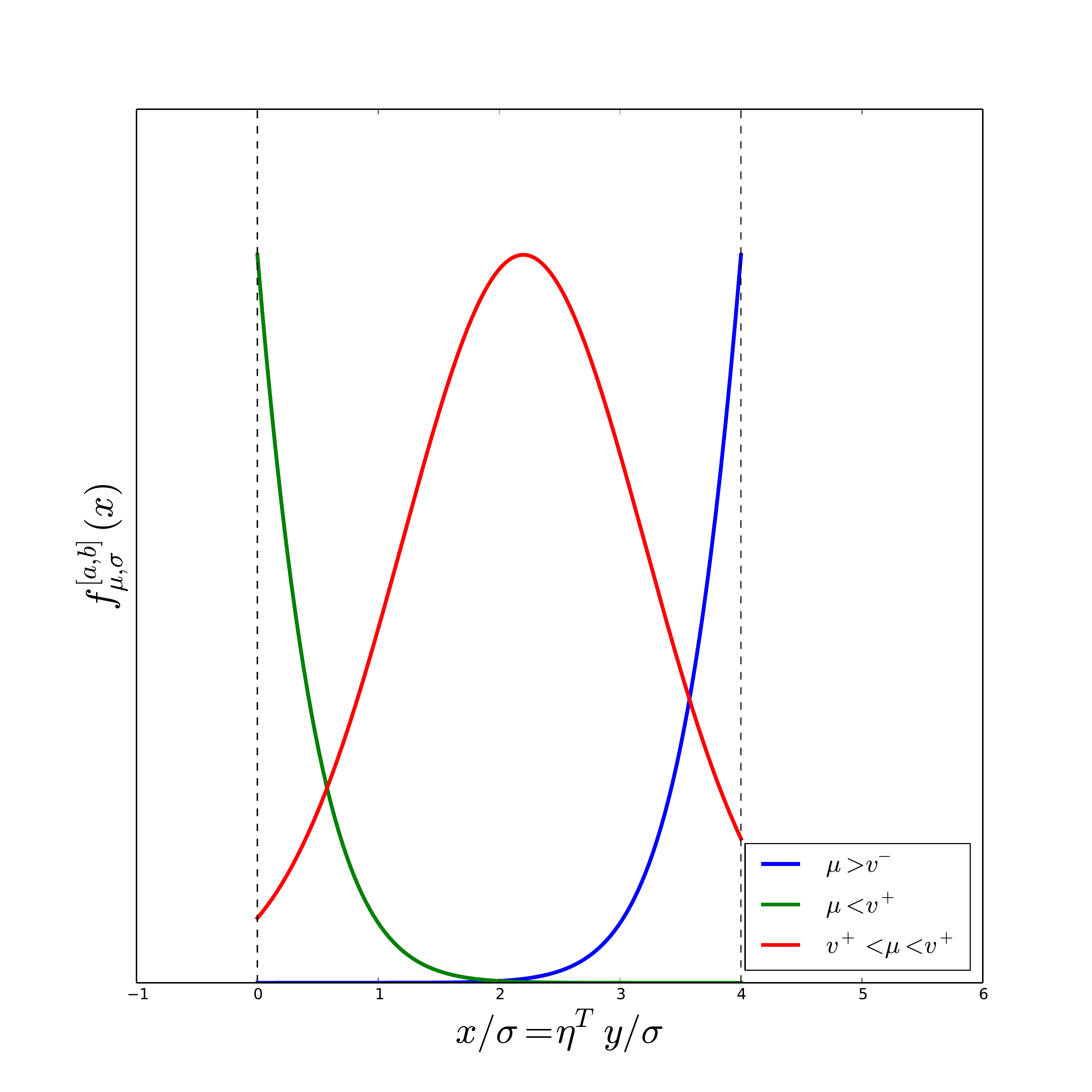}
\caption{The density of the truncated Gaussian $TN(\mu,\sigma^2,a,b)$ depends on the width of $[a,b]$ relative to $\sigma$ as well as the location of $\mu$ relative to $[a,b]$. When
$\mu$ is firmly inside the interval, the distribution resembles a Gaussian. As $\mu$ leaves $[a,b]$, the density begins to converge to an exponential distribution with mean inversely proportional to the distance between $\mu$ and its projection onto $[a,b]$.}
\label{fig:truncated}
\end{figure}

The following pivotal quantity\footnote{The distribution of a pivotal quantity does not depend on unobserved parameters.} follows from Corollary \ref{cor:truncated-normal} via the probability integral transform.

\begin{theorem}
\label{thm:truncated-gaussian-pivot}
Let $\Phi(x)$ denote the CDF of a $N(0,1)$ random variable, and let $F_{\mu, \sigma^2}^{[a, b]}$ denote the CDF of $TN(\mu,\sigma, a,b)$, i.e.:
\begin{equation}
F_{\mu, \sigma^2}^{[a, b]}(x) = \frac{\Phi((x-\mu)/\sigma) - \Phi((a-\mu)/\sigma)}{\Phi((b-\mu)/\sigma) - \Phi((a-\mu)/\sigma)}.
\label{eq:U}
\end{equation}
Then $F_{\eta^T\mu,\ \eta^T \Sigma \eta}^{[\V^-, \V^+]}(\eta^T y)$ is a pivotal quantity, conditional on $\{Ay \leq b\}$:
\begin{equation}
F_{\eta^T\mu,\ \eta^T \Sigma \eta}^{[\V^-, \V^+]}(\eta^T y)\ \big|\ \{Ay \leq b\} \sim \unif(0,1) 
\end{equation}
where $\V^-$ and $\V^+$ are defined in \eqref{eq:v_minus} and \eqref{eq:v_plus}.
\end{theorem}

\begin{figure}[h]
\centering
\includegraphics[width = .45\textwidth]{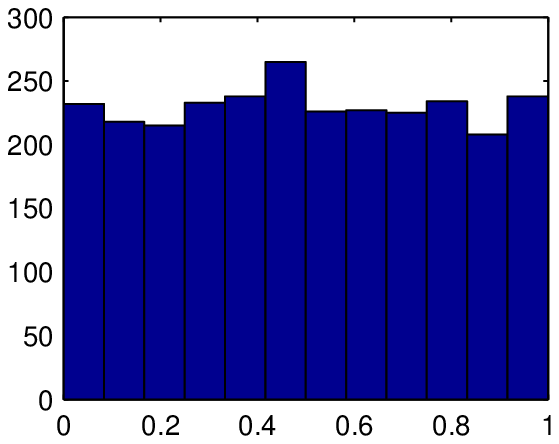}
\includegraphics[width = .45\textwidth]{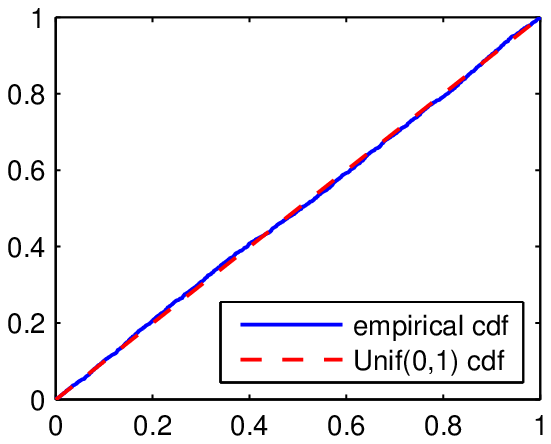}
\caption{Histogram and qq plot of $F_{\eta^T\mu,\ \eta^T \Sigma \eta}^{[\V^-, \V^+]}(\eta^T y)$ where $y$ is a constrained Gaussian. The distribution is very close to $\text{Unif}(0,1)$, which is in agreement with Theorem \ref{thm:truncated-gaussian-pivot}.}
\label{fig:null-qq-plot}
\end{figure}

\section{Inference for marginal screening}
\label{sec:inference-marg-screen}

In this section, we apply the theory summarized in Sections \ref{sec:selection-event} and \ref{sec:truncated-gaussian-test} to marginal screening. In particular, we will construct confidence intervals for the selected variables.

To summarize the developments so far, recall that our model \eqref{eq:model} says that $y \sim N(\mu, \sigma^2 I)$. The distribution of interest is $y| \{\hat E(y) =E\}$, and by Theorem \ref{thm:y-cond-dist}, this is equivalent to $y | {\{A(S, s)z \leq b(S, s) \}}$, where $y \sim N(\mu, \sigma^2 I)$. By applying Theorem \ref{thm:truncated-gaussian-pivot}, we obtain the pivotal quantity
\begin{align}
F_{\eta^T\mu,\ \sigma^2 ||\eta||_2^2}^{[\V^-, \V^+]}(\eta^Ty)\ \big|\ \{\hat E(y)=E\} \sim \unif(0,1)
\label{eq:pivot-Lasso}
\end{align}
for any $\eta$, where $\V^- $ and $\V^+ $ are defined in \eqref{eq:v_minus} and \eqref{eq:v_plus}.  
\subsection{Hypothesis tests for selected variables}

In this section, we describe how to form confidence intervals for the components of $\beta^\star _{\hat S}= (X_{\hat S}^T X_{\hat S})^{-1} X_{\hat S}^T  \mu$. The best linear predictor of $\mu$ that uses only the selected variables is $\beta^\star_{\hat S}$ , and $\hbeta _{\hat S}= (X_{\hat S}^T X_{\hat S})^{-1} X_{\hat S}^T  y$ is an unbiased estimate of $\beta^\star _{\hat S}$. In this section, we propose hypothesis tests and confidence intervals for $\beta^\star_{\hat S}$.
 If we choose 
\begin{equation}
\label{eq:eta_confint}
\eta_j = ((X_{\hat S}^T X_{\hat S})^{-1} X_{\hat S}^T e_j)^T,
\end{equation} 
then $\eta_j^T \mu = \beta_{j\in \hat S}^\star$, so the above framework provides a method for inference about the $j^\text{th}$ variable in the model $\hat S$. This choice of $\eta$ is not fixed before marginal screening selects $\hat S$, but it is measurable with respect to the $\sigma$-algebra generated by the partition. Since it is measurable, $\eta$ is constant on each partition, so the pivot is uniformly distributed on each element of the partition, and thus uniformly distributed for all $y$.

If we assume the linear model $\mu =X\beta^0$ for some $\beta^0\in\reals^p$, $S^0:=\text{support}(\beta^0) \subset \hat S$, and $X_{\hat S}$ is full rank, then by the following computation $\beta^\star_{\hat S} = \beta_{\hat S}^0$:
\begin{align*}
\beta^\star _{\hat S}&=  (X_{\hat S} ^T X_{\hat S})^{-1}X_{\hat S} ^T X_S \beta_S ^0\\
&=  (X_{\hat S} ^T X_{\hat S})^{-1}X_{\hat S} ^T X_{\hat S} \beta_{\hat S} ^0\\
&=\beta_{\hat S} ^0
\end{align*}
In \cite{fan2008sure}, the screening property $S^0 \subset \hat S$ for the marginal screening algorithm is established under mild conditions. Thus under the screening property, our method provides hypothesis tests and confidence intervals for $\beta_{\hat S}^0$.

By applying Theorem \ref{thm:truncated-gaussian-pivot}, we obtain the following (conditional) pivot for $\beta^\star_{j \in \hat S}$:
\begin{align}
F_{\beta^\star_{j \in\hat S},\ \sigma^2 ||\eta_j||^2}^{[\V^-, \V^+]}(\eta_j^Ty)\ \Big| \{\hat E(y)= E\} \sim \unif(0,1).
\label{eq:beta-pivot}
\end{align}
The quantities $j$ and $\eta_j$ are both random through $\hat\E$, a quantity which is fixed after conditioning, therefore Theorem \ref{thm:truncated-gaussian-pivot} holds even for this choice of $\eta$.

Consider testing the hypothesis $H_0: \beta^\star_{j \in \hat S} =\beta_j$. A valid test statistic is given by $F_{\beta_j,\ \sigma^2 ||\eta_j||^2}^{[\V^-, \V^+]}(\eta_j^Ty)$, which is uniformly distributed under the null hypothesis and $y|\{\hat E(y)= E\}$. Thus, this test would reject when $F_{\beta_j,\ \sigma^2 ||\eta_j||^2}^{[\V^-, \V^+]}(\eta_j^Ty) >1-\frac{\alpha}{2}$ or $F_{\beta_j,\ \sigma^2 ||\eta_j||^2}^{[\V^-, \V^+]}(\eta_j^Ty) < \frac{\alpha}{2}$.
\begin{theorem}
The test of $H_0: \beta^\star_{j \in \hat S} =\beta_j$ that accepts when $$\frac{\alpha}{2}<F_{\beta_j,\ \sigma^2 ||\eta_j||^2}^{[\V^-, \V^+]}(\eta_j^Ty)<1-\frac{\alpha}{2}$$ is an $\alpha$ level test of $H_0$.
\end{theorem}
\begin{proof}
Under $H_0$, we have $\beta^\star_{j \in \hat S} =\beta_j$, so by \eqref{eq:beta-pivot} $F_{\beta_j,\ \sigma^2 ||\eta_j||^2}^{[\V^-, \V^+]}(\eta_j^Ty)\big| \{\hat E(y) =E\} $ is uniformly distributed. Thus 
\begin{small}
$$
\Prob( \frac{\alpha}{2}<F_{\beta_j,\ \sigma^2 ||\eta_j||^2}^{[\V^-, \V^+]}(\eta_j^Ty)\le 1-\frac{\alpha}{2} \big| \{\hat E(y) = E,H_0 )\} = 1-\alpha,
$$
\end{small}
and the type 1 error is exactly $\alpha$. Under $H_0$, but not conditional on selection event $\hat E$, we have
\begin{small}
\begin{align*}
&\Prob( \frac{\alpha}{2}<F_{\beta_j,\ \sigma^2 ||\eta_j||^2}^{[\V^-, \V^+]}(\eta_j^Ty)\le 1-\frac{\alpha}{2} \big| H_0 )\}\\
  &=\sum_{\E} \Prob( \frac{\alpha}{2}<F_{\beta_j,\ \sigma^2 ||\eta_j||^2}^{[\V^-, \V^+]}(\eta_j^Ty)\le 1-\frac{\alpha}{2} \big| \{\hat E(y) = E,H_0 )\}\Prob( \hat E(y) =E|H_0)  \\
  &=\sum_E (1-\alpha)   \Prob( \hat E(y) =E|H_0)\\
  &= (1-\alpha)\sum_E \Prob( \hat E(y) =E|H_0)\\
  &=1-\alpha.
\end{align*}
\end{small}
For each element of the partition $\E$, the conditional (on selection) hypothesis test is level $1-\alpha$, so by summing over the partition the unconditional test is level $1-\alpha$.
\end{proof}
Our hypothesis test is not conservative, in the sense that the type 1 error is exactly $\alpha$; also, it is non-asymptotic, since the statement holds for fixed $n$ and $p$. We summarize the hypothesis test in this section in the following algorithm.
\begin{algorithm}
	\caption{Hypothesis test for selected variables}
	\label{alg:test}
\begin{algorithmic}[1]
\State \textbf{Input:} Design matrix $X$, response $y$, model size $k$.
\State Use Algorithm \ref{alg:marg-screen} to select a subset of variables $\hat S$ and signs $\hat s =\text{sign}(X_{\hat S} ^T y)$.
\State Specify the null hypothesis $H_0: \beta^\star_{j \in \hat S}= \beta_j $. 
\State Let $A= A(\hat S,\hat s)$ and $b= b(\hat S,\hat s)$ using \eqref{eq:A-b-defn}. Let $\eta_j = (X_{\hat S}^T)^\dagger e_j$.
\State Compute $F_{\beta_j,\ \sigma^2 ||\eta_j||^2}^{[\V^-, \V^+]}(\eta_j^Ty)$, where $\V^-$ and $\V^+$ are computed via \eqref{eq:v_minus} and \eqref{eq:v_plus} using the $A$, $b$, and $\eta$ previously defined.
\State \textbf{Output:} Reject if  $F_{\beta_j,\ \sigma^2 ||\eta_j||^2}^{[\V^-, \V^+]}(\eta_j^Ty) > \frac{\alpha}{2}$ or  $F_{\beta_j,\ \sigma^2 ||\eta_j||^2}^{[\V^-, \V^+]}(\eta_j^Ty) <1-\frac{\alpha}{2}$.
\end{algorithmic}
\end{algorithm}
\subsection{Confidence intervals for selected variables}
Next, we discuss how to obtain confidence intervals for $ \beta^\star_{j \in \hat S}$. The standard way to obtain an interval is to invert a pivotal quantity \cite{casella1990statistical}. In other words, since
\[ \Prob\left(\frac{\alpha}{2} \leq F_{\beta^\star_{j \in \hat S},\ \sigma^2 ||\eta_j||^2}^{[\V^-, \V^+]}(\eta_j^Ty) \leq 1-\frac{\alpha}{2}\ \big|\ \{\hat E =E\} \right) = \alpha, \]
one can define a $(1-\alpha)$ (conditional) confidence interval for $\beta_{j, \hat\E}^\star$ as 
\begin{align}
\left\{x: \frac{\alpha}{2} \leq F_{x,\ \sigma^2 ||\eta_j||^2}^{[\V^-, \V^+]}(\eta_j^Ty) \leq 1-\frac{\alpha}{2} \right\}. 
\label{eq:conf-int-beta}
\end{align}

In fact, $F$ is monotone decreasing in $x$, so to find its endpoints, one need only solve for the root of a smooth one-dimensional function. The monotonicity is a consequence of the fact that the truncated Gaussian distribution is a natural exponential family and hence has monotone likelihood ratio in $\mu$ \cite{TSH}. 

We now formalize the above observations in the following result, an immediate consequence of Theorem \ref{thm:truncated-gaussian-pivot}.

\begin{corollary}
Let $\eta_j$ be defined as in \eqref{eq:eta_confint}, and let $L_\alpha = L_\alpha(\eta_j,(\hat S,\hat s))$ and $U_\alpha = U_\alpha(\eta_j,(\hat S,\hat s))$ be the (unique) values satisfying 
\begin{align}
F_{L_\alpha,\ \sigma^2 ||\eta_j||^2}^{[\V^-, \V^+]}(\eta_j^Ty) &= 1-\frac{\alpha}{2} & F_{U_\alpha,\ \sigma^2 ||\eta_j||^2}^{[\V^-, \V^+]}(\eta_j^Ty) &= \frac{\alpha}{2} 
\label{eq:L-U-defn}
\end{align}
Then $[L_\alpha, U_\alpha]$ is a $(1-\alpha)$ confidence interval for $\beta^\star _{j \in \hat S}$, conditional on $\hat\E$:
\begin{equation}
\label{eq:coverage}
\Pp \left(\beta^\star _{j \in \hat S} \in [L_{\alpha}, U_{\alpha}]\ \big|\ \{ \hat\E = \E \}  \right) = 1-\alpha.
\end{equation}
\end{corollary}
\begin{proof}
The confidence region of $\beta^\star _{j \in \hat S}$ is the set of $\beta_j$ such that the test of $H_0: \beta^\star _{j \in \hat S}$ accepts at the $1-\alpha$ level. The function $F_{x,\ \sigma^2 ||\eta_j||^2}^{[\V^-, \V^+]}(\eta_j^Ty)$ is monotone in $x$, so solving for $L_\alpha$ and $U_\alpha$ identify the most extreme values where $H_0$ is still accepted. This gives a $1-\alpha$ confidence interval.
\end{proof}
In relation to the literature on False Coverage Rate (FCR) \cite{benjamini2005false}, our procedure also controls the FCR.
\begin{lemma}
For each $j\in \hat S$,
\begin{equation}
\Prob \left( \beta^\star _{j \in \hat S} \in [L_{\alpha}^j,U_{\alpha}^j]  \right) = 1-\alpha.
\end{equation}
Furthermore, the FCR of the intervals $\left\{[L_{\alpha}^j, U_{\alpha}^j]\right\}_{j\in \hat \E}$  is $\alpha$.
\end{lemma}
\begin{proof}
By \eqref{eq:coverage}, the conditional coverage of the confidence intervals are $1-\alpha$. The coverage holds for every element of the partition $\{ \hat E(y) = \E\}$, so
\begin{align*}
&\Prob \left( \beta^\star _{j \in \hat S} \in [L_{\alpha}^j,U_{\alpha}^j]  \right)\\
&= \sum_{\E} \Prob \left(\beta^\star _{j \in \hat S} \in [L_{\alpha}, U_{\alpha}]\ \big|\ \{ \hat\E = \E \}  \right) \Prob(\hat E =E)\\
&= \sum_{\E} (1-\alpha) \Prob (\hat E=E)\\
&= (1-\alpha) \sum_{\E} \Prob (\hat E=E)\\
&=1-\alpha.
\end{align*}
\end{proof}
\begin{figure}[!h]
\centering
\includegraphics[width = .5\textwidth]{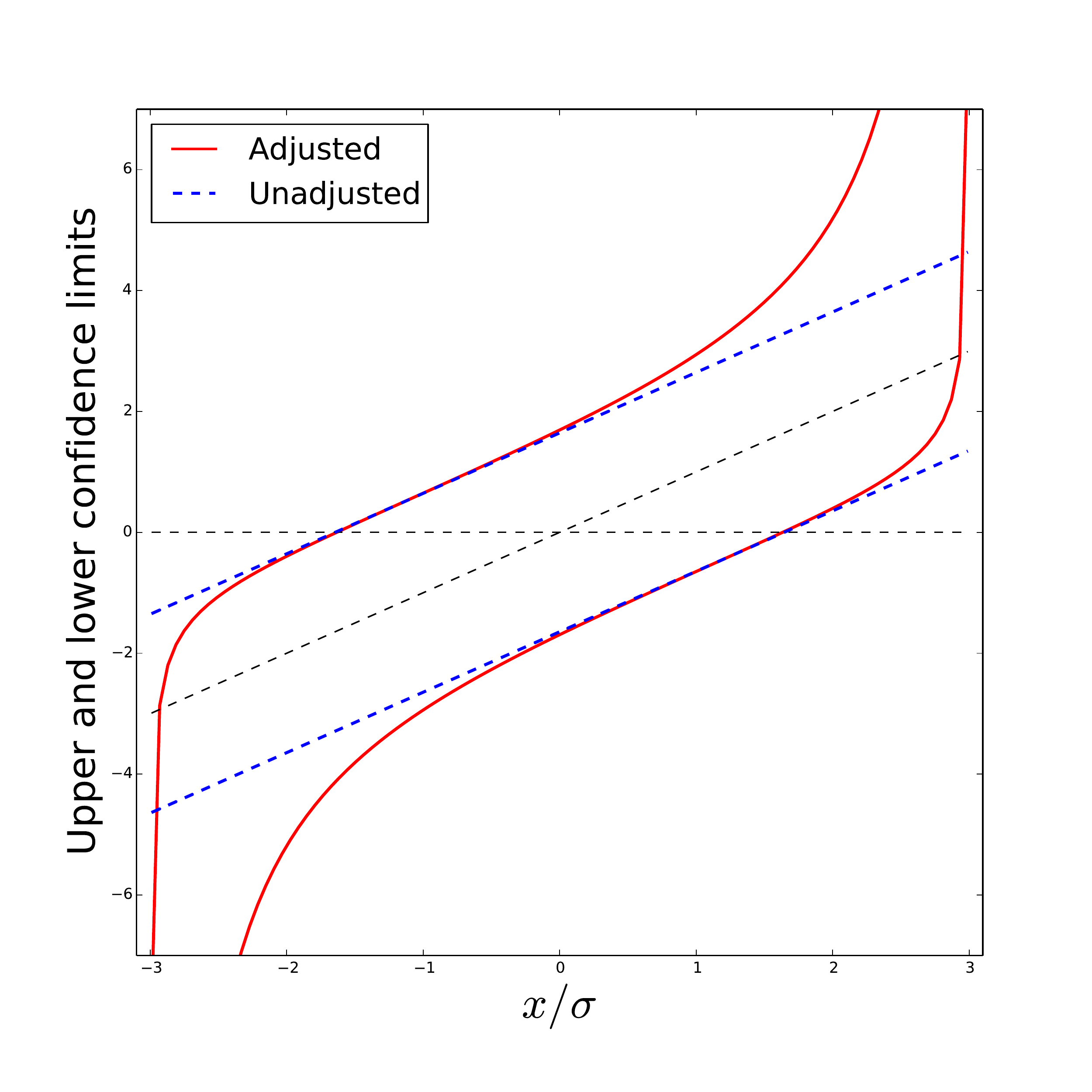}
\caption{Comparison of adjusted and unadjusted 90\% confidence interval for $TN(\mu,\sigma^2,a,b)$. The upper and lower bounds of 90\% confidence intervals are based on $[a,b]=[-3\sigma,3\sigma]$, and the x-axis plots the observation on the scale $\frac{x}{\sigma}$. We see that as long as the obser vation $\frac{x}{\sigma}$ is roughly $0.5\sigma$ away from either boundary, the size of the intervals is comparable to an unadjusted confidence interval. However, the adjusted intervals are guaranteed to have the correct coverage, whereas it is unknown when the unadjusted intervals have the correct coverage.}
\label{fig:intervals}
\end{figure}
We summarize the algorithm for selecting and constructing confidence intervals below.
\begin{algorithm}
	\caption{Confidence intervals for selected variables}
\begin{algorithmic}[1]
\State \textbf{Input:} Design matrix $X$, response $y$, model size $k$.
\State Use Algorithm \ref{alg:marg-screen} to select a subset of variables $\hat S$ and signs $\hat s =\text{sign}(X_{\hat S} ^T y)$.
\State Let $A= A(\hat S,\hat s)$ and $b= b(\hat S,\hat s)$ using \eqref{eq:A-b-defn}. Let $\eta_j = (X_{\hat S}^T)^\dagger e_j$.
\State Solve for $L^j_{\alpha}$ and $U^j_{\alpha}$ using Equation \eqref{eq:L-U-defn} where $\V^-$ and $\V^+$ are computed via \eqref{eq:v_minus} and \eqref{eq:v_plus} using the $A$, $b$, and $\eta_j$ previously defined.
\State \textbf{Output:} Return the intervals $[L^j_{\alpha}, U^j_{\alpha}]$ for $j \in \hat S$.
\end{algorithmic}
\label{alg:ci}
\end{algorithm}

\subsection{Experiments on Diabetes dataset}
\begin{figure}[!h]
\centering
\includegraphics[width = .5\textwidth]{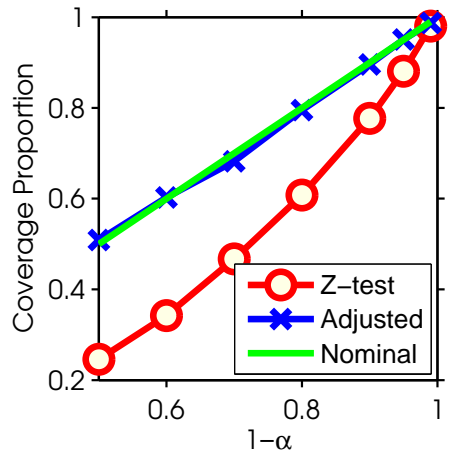}
\caption{Plot of $1-\alpha$ vs the coverage proportion for diabetes dataset. The nominal curve is the line $y=x$. The coverage proportion of the adjusted intervals agree with the nominal coverage level, but the z-test coverage proportion is strictly below the nominal level. The adjusted intervals perform well, despite the noise being non-Gaussian, and $\sigma^2$ unknown.}
\label{fig:diabetes}
\end{figure}
In Figure \ref{fig:fail-z-interval}, we have already seen that the confidence intervals constructed using Algorithm \ref{alg:ci} have exactly $1-\alpha$ coverage proportion. In this section, we perform an experiment on real data where the linear model does not hold, the noise is not Gaussian, and the noise variance is unknown. The diabetes dataset contains $n=442$ diabetes patients measured on $p=10$ baseline variables \cite{efron2004least}. The baseline variables are age, sex, body mass index, average blood pressure, and six blood serum measurements, and the response $y$ is a quantitative measure of disease progression measured one year after the baseline. The goal is to use the baseline variables to predict $y$, the measure of disease progression after one year, and determine which baseline variables are statistically significant for predicting $y$. 

Since the noise variance $\sigma^2$ is unknown, we estimate it by $\sigma^2 = \frac{\norm{y- \hat y}}{n-p}$, where $\hat y = X \hbeta$ and $\hbeta = (X^TX)^{-1} X^T y$. For each trial we generated new responses $\tilde{y}_i = X \hbeta +\tilde \epsilon$, and $\tilde\epsilon$ is bootstrapped from the residuals $r_i = y_i -\hat y_i$. This is known as the residual bootstrap, and is a standard method for assessing statistical procedures when the underlying model is unknown \cite{efron1993introduction}. We used marginal screening to select $k=2$ variables, and then fit linear regression on the selected variables. The adjusted confidence intervals were constructed using Algorithm \ref{alg:ci} with the estimated $\sigma^2$. The nominal coverage level is varied across  $1-\alpha \in \{.5, .6, .7, .8, .9, .95, .99\}$. From Figure \ref{fig:diabetes}, we observe that the adjusted intervals always cover at the nominal level, whereas the z-test is always below. The experiment was repeated $2000$ times.
\section{Extensions}
The purpose of this section is to illustrate the broad applicability of the condition on selection framework. This framework was first proposed in \cite{lee2013exact} to form valid hypothesis tests and confidence intervals after model selection via the Lasso. However, the framework is not restricted to the Lasso, and we have shown how to apply it to marginal screening. For expository purposes, we focused the paper on marginal screening where the framework is particularly easy to understand. In the rest of this section, we show how to apply the framework to marginal screening+Lasso, orthogonal matching pursuit, and non-negative least squares. This is a non-exhaustive list of selection procedures where the condition on selection framework is applicable, but we hope this incomplete list emphasizes the ease of constructing tests and confidence intervals post-model selection via conditioning. 
\label{sec:extensions}
\subsection{Marginal screening + Lasso}
The marginal screening+Lasso procedure was introduced in \cite{fan2008sure} as a variable selection method for the ultra-high dimensional setting of $p=O(e^{n^k})$. Fan et al. \cite{fan2008sure} recommend applying the marginal screening algorithm with $k= n-1$, followed by the Lasso on the selected variables. This is a two-stage procedure, so to properly account for the selection we must encode the selection event of marginal screening followed by Lasso. This can be done by representing the two stage selection as a single event. Let $(\hat S_m, \hat s_m)$ be the variables and signs selected by marginal screening, and the $(\hat S_L, \hat z_L)$ be the variables and signs selected by Lasso \cite{lee2013exact}. In Proposition 2.2 of \cite{lee2013exact}, it is shown how to encode the Lasso selection event $(\hat S_L, \hat z_L)$ as a set of constraints $\{ A_L y \le b_L\}$ \footnote{The Lasso selection event is with respect to the Lasso optimization problem after marginal screening.}, and in Section \ref{sec:selection-event} we showed how to encode the marginal screening selection event $(\hat S_m, \hat s_m)$ as a set of constraints $\{A_m y \le b_m\}$. Thus the selection event of marginal screening+Lasso can be encoded as $\{ A_L y \le b_L, A_m y \le b_m\}$. Using these constraints, the hypothesis test and confidence intervals described in Algorithms \ref{alg:test} and \ref{alg:ci} are valid for marginal screening+Lasso.

\subsection{Orthogonal Matching Pursuit}
Orthogonal matching pursuit (OMP) is a commonly used variable selection method. At each iteration, OMP selects the variable most correlated with the residual $r$, and then recomputes the residual using the residual of least squares using the selected variables. The description of the OMP algorithm is given in Algorithm \ref{alg:omp}.

\begin{algorithm}
	\caption{Orthogonal matching pursuit (OMP)}
\begin{algorithmic}[1]
\State \textbf{Input:} Design matrix $X$, response $y$, and model size $k$.
\State \textbf{for}: $i=1$ to $k$
\State \quad  $p_i = \arg \max_{j=1,\ldots,p} |r_i ^T x_j|$.
\State \quad  $\hat S_i =\cup_{j=1}^i  \ \{p_i\}$.
\State \quad $r_{i+1} = (I- X_{\hat S_i} X_{\hat S_i} ^{\dagger} ) y$.
\State \textbf{end for}
\State \textbf{Output}: $\hat S :=\{p_1, \ldots, p_k\}$, and $\hat \beta_{\hat S} = (X_{\hat S} ^T X_{\hat S} )^{-1} X_{\hat S}^T y$
\end{algorithmic}
\label{alg:omp}
\end{algorithm}

Similar to Section \ref{sec:selection-event}, we can represent the OMP selection event as a set of linear constraints on $y$.
\begin{align*}
\hat E(y) &=\left\{y:\text{sign}(x_{p_i}^T r_i  )x_{p_i}^T r_i > \pm x_{j}^T r_i\text{, for all } j\neq p_i \text{ and all $i \in [k]$}   \right\}\\
&=\{ y:\hat s_i x_{p_i}^T(I-X_{\hat S_{i-1}}X_{\hat S_{i-1}}^\dagger)y  > \pm x_{j}^T(I-X_{\hat S_{i-1}}X_{\hat S_{i-1}}^\dagger)y \text{ and }     \\
& \hat s_i x_{p_i}^T(I-X_{\hat S_{i-1}}X_{\hat S_{i-1}}^\dagger)y>0 \text{, for all }  j\neq p_i \text{, and all $i \in [k]$ }\}\\
&=\left\{y: A(\hat S_1,\ldots, \hat S_k,\hat s_1,\ldots, \hat s_k)  \le b(\hat S_1,\ldots, \hat S_k,\hat s_1,\ldots, \hat s_k)\right\}.
\end{align*}
The selection event encodes that OMP selected a certain variable and the sign of the correlation of that variable with the residual, at steps $1$ to $k$. The primary difference between the OMP selection event and the marginal screening selection event is that the OMP event also describes the order at which the variables were chosen. The marginal screening event only describes that the variable was among the top $k$ most correlated, and not whether a variable was the most correlated or $kth$ most correlated.

Since the selection event can be represented as constraints on $y$, the hypothesis test and confidence intervals described in Algorithms \ref{alg:test} and \ref{alg:ci} are valid for OMP selected $\hat \beta _{\hat S}$.

%

\subsection{Nonnegative Least Squares}
Non-negative least squares (NNLS) is a simple modification of the linear regression estimator with non-negative constraints on $\beta$:
\begin{align}
\arg \min_{\beta: \beta \ge 0} \frac{1}{2} \norm{y-X\beta}^2 .
\label{eq:nnls}
\end{align}
Under a positive eigenvalue conditions on $X$, several authors \cite{slawski2013non,meinshausen2013sign} have shown that NNLS is comprable to the Lasso in terms of prediction and estimation errors. The NNLS estimator also does not have any tuning parameters, since the sign constraint provides a natural form of regularization. NNLS has found applications when modeling non-negative data such as prices, incomes, count data. Non-negativity constraints arise naturally in non-negative matrix factorization, signal deconvolution, spectral analysis, and network tomography; we refer to \cite{chen2009nonnegativity} for a comprehensive survey of the applications of NNLS.

We show how our framework can be used to form exact hypothesis tests and confidence intervals for NNLS estimated coefficients. The primal dual solution pair $(\hat \beta, \hat \lambda)$ is a solution iff the KKT conditions are satisfied, 
\begin{align*}
\hat \lambda_i :=-x_i^T (y-X\hat\beta) &\ge 0 \text{ for all i}\\
\hat \beta &\ge 0.
\end{align*}
Let $\hat S =\{i: -x_i^T (y-X\hat\beta)=0\}$. By complementary slackness $\hat \beta_{-\hat S} =0$, where $-\hat S$ is the complement to the ``active" variables $\hat S$ chosen by NNLS. Given the active set we can solve the KKT equation for the value of $\hat \beta_{\hat S}$,
\begin{align*}
-X_{\hat S}^T (y- X\hat \beta) =0\\
-X_{\hat S} ^T (y- X_{\hat S} \hat \beta{\hat S}) =0\\
\hat \beta_{\hat S} = X_{\hat S} ^\dagger y,
\end{align*}
which is a linear contrast of $y$. The NNLS selection event is 
\begin{align*}
\hat E(y)&=\{y: X_{\hat S} ^T (y-X\hat \beta) =0,\ X_{-\hat S}^T (y-X\hat \beta) >0\}\\
&=\{y: X_{\hat S}^T (y-X\hat \beta) \ge 0, -X_{\hat S}^T (y-X \hat \beta) \ge 0, X_{-\hat S}^T (y-X\hat \beta) >0\}\\
&=\{y: X_{\hat S}^T (I- X_{\hat S} X_{\hat S}^\dagger)y \ge 0, -X_{\hat S}^T(I- X_{\hat S} X_{\hat S}^\dagger)y \ge 0, X_{-\hat S}^T (I- X_{\hat S} X_{\hat S}^\dagger)y >0\}\\
&=\{y: A(\hat S) y \le 0\}.
\end{align*}
The selection event encodes that for a given $y$ the NNLS optimization program will select a subset of variables $\hat S(y)$. Similar to the case in OMP and marginal screening, we can use Algorithms \ref{alg:test} and \ref{alg:ci}, since the selection event is represented by a set of linear constraints $\{y: A(\hat S) y \le 0\}$.

\section{Conclusion}
Due to the increasing size of datasets, marginal screening has become an important method for fast variable selection. However, the standard hypothesis tests and confidence intervals used in linear regression are invalid after using marginal screening to select important variables. We have described a method to perform hypothesis and form confidence intervals after marginal screening. The conditional on selection framework is not restricted to marginal screening, and also applies to OMP, marginal screening + Lasso, and NNLS. 
\section*{Acknowledgements}
Jonathan Taylor was supported in part by NSF grant DMS 1208857 and AFOSR grant 113039. Jason Lee was supported by a NSF graduate fellowship, and a Stanford Graduate Fellowship.
\newpage
\bibliographystyle{plain}
\bibliography{marginal_screening}

\end{document}